\documentclass{llncs}

\usepackage{amssymb}   

\usepackage{xspace}   
\usepackage{color}    

\DeclareRobustCommand{\reFLect}{\textit{re\kern-0.07em F\kern-0.07emL\kern-0.29em\raisebox{0.56ex}{ect}}\xspace}
\DeclareRobustCommand{\ReFLect}{\textit{Re\kern-0.07em F\kern-0.07emL\kern-0.29em\raisebox{0.56ex}{ect}}\xspace}

\newcommand{\tree}[2]{{\renewcommand{\arraystretch}{1.1}%
{\begin{tabular}{c}#1 \\ \hline #2\end{tabular}}}}

\providecommand{\type}[1]{\ensuremath{\mathord{\textsl{#1}}}}
\renewcommand{\type}[1]{\ensuremath{\mathord{\textsl{#1}}}}

\newcommand{\term}{\type{term}}
\newcommand{\unit}{\type{unit}}
\newcommand{\bool}{\type{bool}}

\newcommand{\sterm}{\textsl{\hbox{\scriptsize term}}}

\newcommand{\sbool}{\textsl{\hbox{\scriptsize bool}}}

\newcommand{\pterm}[1]{(#1)\type{term}}

\newcommand{\lquote}{\mathopen{\langle}}
\newcommand{\rquote}{\mathclose{\rangle}}
\newcommand{\quot}[1]{\lquote #1 \rquote}
\newcommand{\const}[1]{\ensuremath{\mathord{\mathsf{#1}}}}

\def\sem#1{{[\mkern-8mu[}\,#1\,{]\mkern-8mu]\mkern2mu}}

\def\set#1{\mathit{#1}}

\begin{document}


\title{On the Semantics of ReFLect\\
       as a Basis for a Reflective Theorem Prover} 

\date{\today}

\author{Tom Melham\inst{1} \and 
        Raphael Cohn\inst{2} \and 
        Ian Childs\inst{1}}

\institute{University of Oxford \and The Rockefeller University}


\maketitle

\thispagestyle{plain}

\begin{abstract}
This paper explores the semantics of a combinatory fragment of \reFLect,
the $\lambda$-calculus underlying a functional language used by Intel
Corporation for hardware design and verification.  \ReFLect is similar to
ML, but has a primitive data type whose elements are the abstract syntax
trees of \reFLect expressions themselves.  Following the LCF paradigm, this
is intended to serve as the object language of a higher-order logic theorem
prover for specification and reasoning---but one in which object- and
meta-languages are unified. The aim is to intermix program evaluation and
logical deduction through reflection mechanisms. We identify some
difficulties with the semantics of \reFLect as currently defined, and
propose a minimal modification of the type system that avoids these
problems.
\end{abstract}

\section{Introduction}
	
\ReFLect is a strongly typed, functional programming language, similar to
ML, with certain reflection features for applications in hardware design
and verification~\cite{Grundy:2006:RFL}. The language was designed by
researchers at Intel Corporation and forms the basis for Intel's Forte
formal verification environment~\cite{Seger:2005:IEE}.  Forte has been used
at Intel to attack many challenging verification problems for real-world
processor designs; one impressive achievement is the verification,
without simulation, of much of the Core i7 processor execution
cluster~\cite{DBLP:conf/cav/KaivolaGNTWPSTFRN09}.

Reflection is supported through a primitive data type, \term, whose values
are \reFLect abstract syntax trees.  A \emph{quotation} $\quot{e}$ has type
\term\ and denotes the syntax tree of the term $e$. In Forte, this type is
the basis for a higher-order logic theorem prover that is similar to
HOL~\cite{Gordon:1993:IHA}. In systems like HOL, the higher-order logic
`object language' is built on top of the $\lambda$-calculus, following
Church~\cite{Church:1940:AFS}.  The syntax of the logic is represented by
an algebraic data-type in a functional programming `meta-language'.
\ReFLect is designed to unify the object and meta languages in this kind of
enterprise. The aim is for the \emph{same} $\lambda$-calculus to be the
core of both logic and meta-language.  Moreover, it should be possible to
move freely between \emph{evaluation} in the interpreter and
\emph{deduction} in the theorem prover. For example, the reduction rules of
\reFLect align with inference rules of its logic, so proof can be done just
by evaluation: to prove a theorem $\quot{P}$, simply strip off the quotes
and check that the interpreter evaluates $P$ to `true'.  
		
The logical soundness of this depends on having the right formal semantics
to justify these reflection rules. There exists an operational
(i.e.~reduction) semantics for full \reFLect, including quotation
evaluation~\cite{Grundy:2006:RFL}, but for logic a denotational semantics
is needed. Krsti\'{c} and Matthews have published
one~\cite{DBLP:conf/ppdp/KrsticM04}, but this omits the crucial reflection
constructs that bridge logic and evaulation.

In this paper, we analyse the semantics of a simplified system,
\textit{Combinatory} \ReFLect, that has some core features of full \reFLect
and includes the reflection constructs missing from the semantics
in~\cite{DBLP:conf/ppdp/KrsticM04}. This gives a simple setting to
investigate the problem while side-stepping the technicalities of variable
binding in the full language.  We find that even our variable-free language
does not support the semantics we need, and that any reasonable logic built
on \reFLect will be inconsistent. Our proposed solution is a modification
of the language's type system that stratifies the semantics enough to avoid
the problem.

\section{Combinators with Reflection}\label{cr}

\emph{Combinatory} \ReFLect (CR) is a variable-free system of combinatory
logic that includes essentially the same reflection constructs, quotations,
and typing system as full \reFLect.  Indeed, CR can be viewed as a
sublanguage of full \reFLect.  The syntax of CR includes explicitly-typed
versions of the combinators $\const{I}$, $\const{K}$ and $\const{S}$, along
with reflective operators $\const{value}$, $\const{app}$ and
$\const{lift}$. CR also supports quotation of any term $e$ in the language,
written $\quot{e}$.

The types of CR terms, $\set{Ty}$, are defined inductively by
\begin{displaymath} 
\sigma\;\; {:}{=} \;\; \bool \; \mid \; \unit \; \mid \;  \term  \;  \mid \; \sigma_1 \rightarrow \sigma_2
\end{displaymath} 

\noindent And the terms of CR, $\set{Exp}$, are given by
\begin{equation}\label{syntax}  
   e \;\; {:}{=} \;\; \const{I}_\sigma \;  
       	\mid  \; \const{K}_{\sigma,\tau} \;
        \mid  \; \const{S}_{\sigma,\tau,\upsilon}  \;
       	\mid  \; \const{value}_\sigma \;
	\mid  \; \const{lift} \;
	\mid \; \const{app} \;
       	\mid  \; e_1\; e_2 \;
       	\mid  \; \quot{e}
\end{equation} 

\noindent where $\sigma, \tau$ and $\upsilon$ range over types.  We refer
to terms of the form `$e_1 \; e_2$' as \emph{applications}, and terms of
the form `$\quot{e}$' as \emph{quotations}.

A quotation is an object-language phrase of CR that, semantically, denotes
the abstract syntax tree of the term inside the quotes. A quotation
$\quot{e}$ is semantically different from the term $e$. Unquoted terms that
evaluate to the same value, for example $(\const{S}\:\const{K})\:\const{K}\;e$
and $\const{I}\;e$, are semantically equal. But
$\quot{(\const{S}\:\const{K})\:\const{K}\;e}$ and $\quot{\const{I}\;e}$ are
semantically distinct; they denote different syntax trees.

To explain the intended semantics of $\const{value}_\sigma$,
$\const{lift}$, and $\const{app}$, we suppose, just for a moment, that
normal forms exist in some reduction system for CR. Of course we have not
\emph{established} this yet. The intended semantics, then, is that
$\const{value}_\sigma\:\quot{e}$ should denote the normal form or
evaluation of $e$ of type $\sigma$. This is obviously problematic if $e$
contains free variables, so the operational semantics of \reFLect allows
reduction of $\const{value}_\sigma\:\quot{e}$ only for closed
$e$~\cite[p. 187]{Grundy:2006:RFL}.  The $\const{lift}$ function reifies
values into representative syntax and, in full \reFLect, applies only to
closed terms that possess a canonical `name' for their value~\cite[\S
  8.2]{Grundy:2006:RFL}.  For example, both $\const{lift}\;2$ and
$\const{lift}\;(1+1)$ reduce to $\quot{2}$. In CR, we adopt an even more
conservative semantics for \const{lift}: it applies only to terms of type
$\term$, and we can reduce only terms of the form
$\const{lift}\;\quot{e}$. It seems obviously harmless to suppose we might
have a function that takes $\quot{e}$ to $\quot{\quot{e}}$.  Finally, in
what follows we will need to apply one $\term$ to another; this is native
in full \reFLect, but in CR we have an ad-hoc combinator, $\const{app}$.

\begin{figure}[t]
\begin{center}
\tree{}{$\const{I}_\sigma: \sigma \rightarrow \sigma$}
  \quad
\tree{}{$\const{K}_{\sigma,\tau} : \sigma \rightarrow \tau \rightarrow \sigma$}
  \quad
\tree{}{$\const{S}_{\sigma,\tau,\upsilon}:(\sigma \rightarrow \tau \rightarrow \upsilon) \rightarrow (\sigma \rightarrow \tau) \rightarrow \sigma \rightarrow \upsilon$}\\
\tree{}{$\const{value}_\sigma: \term \rightarrow \sigma$}
  \quad
\tree{}{$\const{lift}: \term \rightarrow \term$} 
  \quad
\tree{}{$\const{app} : \term \rightarrow \term \rightarrow \term$}\\[2.5mm]
\tree{$e_1 : \sigma \rightarrow \tau \quad e_2: \sigma$}{$e_1 \; e_2 : \tau$}
  \qquad
\tree{$e : \sigma$}{$\quot{e} : \term$}
\end{center}
\caption{Typing rules for CR.}\label{typing}
\end{figure}

Figure \ref{typing} shows the type system of CR.  The judgement $e :
\sigma$ means that the term $e$ has type $\sigma$.  Note that the type
system does not rule out the formation of a term
`$\const{value}_{\sigma}\:\quot{e}$' where $e$ does not have type
$\sigma$. A check for this kind of type mismatch is done only at runtime in
CR, during reduction. In Figure~\ref{reductions} we show part of the
reduction system for CR.  (We omit only the rules that close $\Rightarrow$ under
reflexivity, transitivity, and congruence.)
	
\begin{figure}[t]
\begin{displaymath}
\begin{array}[]{@{}c@{}}
\const{I}_\sigma \; e \Rightarrow e \qquad\quad
\const{K}_{\sigma,\tau} \; e_1 \; e_2 \Rightarrow e_1 \qquad\quad
\const{S}_{\sigma,\tau,\upsilon} \; e_1 \; e_2 \; e_3 \Rightarrow e_1 \; e_3 \; (e_2 \; e_3) \\[3mm]
\textrm{if $e : \sigma$, then\ }\const{value}_\sigma \; \quot{e} \Rightarrow e \qquad\quad
 \const{lift} \; \quot{s} \Rightarrow \quot{\quot{s}} \\[3mm]
\textrm{if $e_1 \; e_2$ is well typed, then\ }\const{app} \; \quot{e_1}\; \quot{e_2} \Rightarrow \quot{e_1 \; e_2} \\[3mm]
\end{array}
\end{displaymath}
\caption{Some Reductions in CR.}\label{reductions}
\end{figure}

\subsection{Denotational Semantics}

The reflective operators $\const{value}_\sigma$ and $\const{lift}$ are key
components of the \reFLect language. They serve as the essential bridge
betwen deductive logic and program evaluation in Forte. It is these
operators, omitted from the semantics in~\cite{DBLP:conf/ppdp/KrsticM04},
that make the denotational semantics of \reFLect difficult, as we now show.

Suppose we have a semantic function, $\sem{-} \in \set{Exp} \rightarrow
\set{U}$ that maps each term to an element of some universe $\set{U}$.  A
quotation $\quot{e}$ represents the abstract syntax tree of its contents
$e$. So the formal denotation of a quoted term $\sem{\quot{e}}$, where $e$
is a closed term, can just be the \reFLect term $e$ itself, an element of
$\set{Exp} \subseteq \set{U}$. And, indeed, this is how it is defined
in~\cite{DBLP:conf/ppdp/KrsticM04}.  With the usual semantics of function
application, we would therefore expect for closed $e : \sigma$ that
\begin{displaymath}
\sem{\const{value}_\sigma \; \quot{e}} = 
\sem{\const{value}_\sigma}\;\sem{\quot{e}} = 
\sem{\const{value}_\sigma}\; e = \sem{e}
\end{displaymath}

\noindent So the action of the $\const{value}_\sigma$ function seems to
coincide with part of $\sem{-}$ itself.

It has been shown that core \reFLect, without the reflective functions, is
confluent and normalizing~\cite{DBLP:conf/ppdp/KrsticM04}. The authors
of~\cite{DBLP:conf/ppdp/KrsticM04} used this normalizing result as a
pre-requisite for proving the soundness of their denotational semantics. We
now show that adding the reflective functions $\const{lift}$ and
$\const{value}_\sigma$ destroys this property.
	
\begin{theorem}\label{confluence}
Combinatory \reFLect is not strongly normalizing.
\end{theorem}

\begin{proof}
We proceed by constructing a term that has an infinite reduction
sequence. We first define a CR expression $f$ such that for all $e$ of type
$\term \rightarrow \term$, $f\;\quot{e} \Rightarrow \quot{e\:\quot{e}}$.
In full \reFLect, $f$ would be $\lambda x.\, \const{app} \;x
\;(\const{lift}\;x)$, which in combinators is
$\const{S}_{\sterm,\sterm,\sterm} \; \const{app}\; \const{lift}$. It is
easy to check that this term is typeable with type $\term \rightarrow
\term$ by the rules in Figure~\ref{typing}. We also have the reduction
sequence
\begin{displaymath}
f\;\quot{e} = 
\const{S} \; \const{app}\; \const{lift} \; \quot{e} 
\Rightarrow 
\const{app}\; \quot{e} \; (\const{lift} \; \quot{e})
\Rightarrow 
\const{app}\; \quot{e} \; \quot{\quot{e}}
\Rightarrow 
\quot{e\: \quot{e}}
\end{displaymath}

\noindent for any $e$ of type $\term \rightarrow \term$. (We omit type
subscripts for readability.)

Now, let the term $g$ of type $\term \rightarrow \term$ be
$\const{S}\;(\const{K}\;\const{value_{\sterm}})\;
(\const{S}\;(\const{K}\;f)\;\const{I})$; for readability we omit type
subscripts.  In full \reFLect, $g$ would be $\lambda x.\,
\const{value}_{\sterm} \; (f \; x)$.  We immediately see that $g\:\quot{g}$
is well typed. But now we have a circular reduction sequence: $g \;
\quot{g} \Rightarrow \const{value_{\sterm}}\; (f\;\quot{g}) \Rightarrow
\const{value_{\sterm}}\; (\quot{g\:\quot{g}}) \Rightarrow g\:\quot{g}$.

\end{proof}
	
\subsection{Indefinability in CR}

We now show that any higher order logic built on \reFLect and containing
the sublanguage CR will be inconsistent.  We will proceed by supposing we
do have such a logic, formulated in the usual way deriving from
Church~\cite{Church:1940:AFS}. That is, we suppose a logic has been defined
on top of CR in the same way that the HOL logic is defined on top of the
simply-typed $\lambda$-calculus~\cite{Gordon:1993:IHA}.  We then
demonstrate inconsistency by a construction inspired by Tarski's
Indefinability Theorem \cite{Tarski31}.
	
\begin{theorem}\label{diag}
In a conventional higher-order logic built on CR, for any
term $\Psi: \term \rightarrow \bool$, there is a typeable expression,
$\Gamma$, such that $\vdash \Gamma = \Psi \quot{\Gamma}$.
\end{theorem}
	
\begin{proof}
Suppose $\Psi : \term \rightarrow \bool$.  Define $\beta$ of type $\term
\rightarrow \bool$ such that for any $e$ of type $\term$, $\beta \; e
\Rightarrow \Psi (f\;e)$, where $f =\const{S}_{\sterm,\sterm,\sterm} \; \const{app}\; \const{lift}$,
as in the proof of Theorem~\ref{confluence}.  In full \reFLect, $\beta$ would just be $\lambda
x.\;\Psi(f\;x)$, where $x$ is chosen not to occur free in $\Psi$. We now
let $\Gamma$ be $\beta\;\quot{\beta}$ and prove $\vdash \Gamma = \Psi \;
\quot{\Gamma}$ in our assumed higher order logic:
\[
\Gamma\; \begin{array}[t]{@{}cl@{\qquad\quad-}l}
= & \beta\;\quot{\beta} & 
  \textrm{definition of $\Gamma$}\\
= & \Psi(f\;\quot{\beta})&
  \textrm{reduction of $\beta\;\quot{\beta}$}\\
= & \Psi\;\quot{\beta\;\quot{\beta}} &
  \textrm{reduction of $f\;\quot{\beta}$}\\ 
= & \Psi\;\quot{\Gamma} & 
  \textrm{definition of $\Gamma$}
\end{array}\]

\end{proof}

\begin{corollary}\label{inconsistent} 
Any conventional higher order logic built on CR is inconsistent.
\end{corollary}
	
\begin{proof} 
The proof is straightforward and we sketch it here. As discussed, CR
contains a truth predicate, $\const{value}_{\sbool}$. We then have a
falsity predicate \const{isfalse} defined to be
$\const{S}\;(\const{K}\;\neg)\;(\const{S}\;(\const{K}\;\const{value}_{\sbool})\;\const{I})$.
Taking $\Psi$ in Theorem~\ref{diag} to be   $\const{isfalse}$, we conclude that any 
logic built on CR can prove $\vdash \Gamma = \const{isfalse} \;
\quot{\Gamma}$. But then
$\const{value}_{\sbool} \; \quot{\Gamma} = \Gamma = \const{isfalse} \; \quot{\Gamma} = \neg \; (\const{value}_{\sbool} \; \quot{\Gamma})$, a contradiction.
\end{proof}
	
\section{Stratified Typing for Terms}\label{SCR}

Analysis of the above results for CR reveals that its fundamental flaw is
that all quotations have a single type, $\term$. This allows the
circularities that prevent normalization and make it logically
inconsistent.  We now sketch a variant, \emph{Stratified Combinatory
  \ReFLect} (SCR), that avoids these pitfalls.  The basic idea is simple:
the type of a quotation will carry with it the type of the term inside.
	
The syntax and reduction semantics of SCR are essentially the same as
in CR. The difference is in the types.  The types of SCR
terms are defined by
\begin{displaymath}
\sigma\;\; {:}{=}\;\; \unit \; \mid \;  \pterm{\sigma} \; \mid \; \sigma_1 \rightarrow \sigma_2
\end{displaymath} 
	
\noindent and the typing rules for the reflective combinators and
quotations are
\vskip-4mm
\begin{displaymath}
\begin{array}{@{}c@{}}
\tree{}{$\const{value}_\sigma: \pterm{\sigma} \rightarrow \sigma$}
 \qquad
\tree{}{$\const{lift}_\sigma: \pterm{\sigma} \rightarrow \pterm{\pterm{\sigma}}$}\\[1.5mm]
\tree{}{$\const{app}: \pterm{\sigma \rightarrow \tau} \rightarrow 
    \pterm{\sigma} \rightarrow \pterm{\tau}$}\qquad
\tree{$e : \sigma$}{$\quot{e} : \pterm{\sigma}$}
\end{array}
\end{displaymath}

\noindent The other combinators and applications are typed as they are in CR.

This new typing schemes rules out applying a function to its quoted
self. It bans the self-application used to define the nonterminating
`$g\:\quot{g}$' in Theorem~\ref{confluence}, and it rules out the `$\beta\;
\quot{\beta}$' construction used in the proof of Theorem~\ref{diag}.

\section{Conclusions and Discussion}

This paper has explored the semantics of a reflective combinatory logic
that shares key features with \reFLect, a functional language intended to
unify computation and deduction in a practical engineering setting. Our
analysis suggests that the type of quotations in \reFLect must be
parameterized by the type of their contents.  We speculate that, with this
adjustment to the type system, a set-theoretic semantics of full \reFLect,
including reflection, will be possible.

To make \ReFLect attractive to verification engineers, it has a simple
Hindley-Milner type system. This means quotations cannot have types of the
form $(\sigma)\term$ without making the definitions of certain common
functions over terms untypeable. For example, we cannot define
$\const{operand}\;\quot{e_1\; e_2} = \quot{e_1}$, since this function would
have to have an existential type $(\alpha)\term \rightarrow \exists
\beta.\:(\beta\rightarrow\alpha)\term$.  Also problematic are functions
defined recursively over the syntax of terms, which are ubiquitous in
theorem-prover code.

Some developments in functional programming, subsequent to the design of
\reFLect, offer a way forward. \emph{Generalized algebraic data types}
(GADTs) are a generalization of standard algebraic data types that take a
modest step towards dependent types~\cite{DBLP:conf/icfp/JonesVWW06}.
GADTs allow for algebraic data type constructors to have parameters that
can be instantiated to specific types within the body of a function defined
over values of the type. We speculate that one could treat quotations and
the type $\pterm{\sigma}$ as a GADT in \reFLect, with the aim of making it
possible to define the term-traversing functions needed to implement a
theorem prover.

\ReFLect currently includes neither GADTs nor $\pterm{\sigma}$.  A major
project for the future is to develop a full reflective language with GADTs
and a theorem prover based on the semantic insights in this paper.  A clear
first step is to investigate whether GADTs, as found for example in
Haskell, are indeed sufficient develop a theorem prover with a
parameterized type $\pterm{\sigma}$ of terms. We have done an inital
investigation of this, and the answer seems to be `partly'; the main
difficulty seems to be finding a satisfactory treatment of polymorphism.

We are of course aware that there are functional languges and logical
calculi with more flexible type systems, among them System F, Coq, and
HOL-Omega. But in this industrially-motivated work we have been exploring
options that give practicing engineers as simple and intuitive a functional
programming language as possible, and so aim to remain close to the
Hindley-Milner type system.

More generally, fast object-language evaluation has been a goal of theorem
prover designers since the earliest days. Approaches include term data
structures that optimise symbolic evaluation by proof, and extraction of
programs in standard functional languages~\cite{Coq}. The work presented
here is distinguished in attempting direct unification of object- and
meta-languages, of symbolic reasoning and direct program execution.

\medskip

\noindent\textit{Acknowledgment} This work leans heavily on a preliminary
sketch of the semantics of full \reFLect devised in collaboration with Jim
Grundy and Sava Krsti\'{c}~\cite{Grundy07}

\bibliographystyle{abbrv}
\bibliography{paper} 

\end{document}